\documentclass[12pt]{article}
\usepackage{fullpage}

\usepackage{amsthm,amsmath,txfonts}
\usepackage[T1]{fontenc}

\newcommand{\pw}{\textit{pw}} 
\newcommand{\ppw}{\textit{ppw}} 
\newcommand{\vbw}{\textit{vbw}} 
\newcommand{\bw}{\textit{bw}} 

\newcommand{\wei}{\textit{wei}} 

\newcommand{\Bndry}{\partial} 
\newcommand{\nBndry}{\beta}   

\newcommand{\ls}{\mathbin{\prec}}
\newcommand{\gs}{\mathbin{\succ}}
\newcommand{\les}{\mathbin{\preceq}}
\newcommand{\ges}{\mathbin{\succeq}}

\newcommand{\cp}{\mathbin{\square}}

\newcommand{\dGrid}[1]{\prod_{i=1}^{#1} P_{n_{i}}} 
\newcommand{\dTorus}[1]{\prod_{i=1}^{#1} C_{2n_{i}}} 

\newcommand{\unit}[1]{\hat{#1}}

\newcommand{\floor}[1]{\left\lfloor #1 \right\rfloor}

\newcommand{\IF}{\textrm{if }} 
\newcommand{\OW}{\textrm{otherwise}} 

\newcommand{\theoremname}{Theorem}
\newtheorem{thm}{\theoremname}[section]

\newcommand{\lemname}{Lemma}
\newtheorem{lem}[thm]{\lemname}

\newcommand{\corname}{Corollary}

\newtheorem{cor}[thm]{\corname}

\newcommand{\propname}{Proposition}

\newcommand{\factname}{Fact}
\newtheorem{fact}[thm]{\factname}

\newcommand{\claimname}{Claim}
\newtheorem{claim}[thm]{\claimname}

\newcommand{\conjname}{Conjecture}
\newtheorem{conj}[thm]{\conjname}

\begin{document}
 \title{Bandwidth and pathwidth of three-dimensional grids}

 \author{Yota Otachi\thanks{
 Graduate School of Information Sciences, Tohoku University,
 Sendai 980-8579, Japan. JSPS Research Fellow.
 Corresponding author. \texttt{otachi@comp.cs.gunma-u.ac.jp}
 }
\and
Ryohei Suda\thanks{
Department of Computer Science, Gunma University,
 1-5-1 Tenjin-cho, Kiryu, Gunma 376-8515, Japan.
}
}
 %


 \maketitle
 
 \begin{abstract}
  We study the bandwidth and the pathwidth of multi-dimensional grids.
  It can be shown for grids, that these two parameters are equal to a more basic graph parameter, the vertex boundary width.
  Using this fact, 
  we determine the bandwidth and the pathwidth of three-dimensional grids, which were known only for the cubic case.
  As a by-product, we also determine the two parameters of multi-dimensional grids with relatively large maximum factors.
 \end{abstract}


 \section{Introduction}
 \label{sec:intro}
 In this paper, we study two well-known graph parameters, the bandwidth and the pathwidth.
 These two parameters were defined in different areas of Computer Science.
 However, there is a close relation between the two parameters.
 In fact, it is known that the bandwidth of a graph is at least its pathwidth.
 Furthermore, as we will show, these two parameters are identical for grids.
 Since grid-like graphs, especially two or three-dimensional grids, arise in many practical situations,
 several graph parameters of them have been studied intensively~\cite{RSV95,SSV04,TV2010,BDH04,KOY09sec,BaloghMubayiPluhar2006,PikhurkoWojciechowski2006}.
 In particular, the bandwidth and the pathwidth of grids and tori were studied by several researchers~\cite{Fit74,Chv75,PSSV93,EW08}.
 However, closed formulas of these parameter for noncubic three-dimensional grids and four or more-dimensional grids were not known previously.
 We study these grids and present closed formulas for some cases.

 The rest of this paper is organized as follows.
 In Section~\ref{sec:pre}, 
 we give some definitions and known results.
 In Section~\ref{sec:large},
 we determine the two parameters for multi-dimensional grids that have relatively large maximum factors.
 Generally speaking, large-dimensional cases are difficult to handle.
 However, we show that if the maximum factor in a grid is relatively large then the two parameters can be easily determined.
 In Section~\ref{sec:3gri},
 we determine the two parameters of three-dimensional grids.
 FitzGerald~\cite{Fit74} determined the bandwidth of cubic grids $P_{n} \cp P_{n} \cp P_{n}$.
 We properly extend this result to noncubic cases $P_{n_{1}} \cp P_{n_{2}} \cp P_{n_{3}}$.
 In the last section, we conclude this paper and give a conjecture for the four-dimensional case.


 \section{Preliminaries}
 \label{sec:pre}
 In this section, we define some graph parameters, and a graph operation.
 Grids, and tori are also defined here.
 After the definitions, we provide some known results as well as some useful observations.
 All graphs in this paper are finite, simple, and connected.
 We denote the vertex and edge sets of a graph $G$ by $V(G)$ and $E(G)$, respectively.

  \subsection{Bandwidth, pathwidth, and vertex boundary width}
  The bandwidth of graphs was defined by Harper~\cite{Har66}.
  An \emph{ordering} of a graph $G$ is a bijection $f \colon V(G) \to \{1,2,\dots,|V(G)|\}$.
  The \emph{bandwidth of an ordering $f$} is $\bw_{G}(f) = \max_{\{u,v\} \in E(G)} |f(u) - f(v)|$,
  and the \emph{bandwidth} of $G$ is $\bw(G) = \min_{f} \bw_{G}(f)$.
  The bandwidth problem appears in a lot of areas of Computer Science such as
  VLSI layouts and parallel computing. See surveys~\cite{CCDG82,DPS02}.

  The pathwidth of graphs was defined by Robertson and Seymour~\cite{RS83} in their work of the Graph Minor Theory.
  Given a graph $G$, a sequence $X_{1}, \dots, X_{r}$ of subsets of $V(G)$ is a \emph{path decomposition} of $G$
  if the following conditions are satisfied:
  \begin{enumerate}
   \item $\bigcup_{1 \le i \le r} X_{i} = V(G)$,
   \item for each $\{u,v\} \in E(G)$, there exists an index $i$ such that $\{u,v\} \subseteq X_{i}$,
   \item for $i \le j \le k$, $X_{i} \cap X_{k} \subseteq X_{j}$.
  \end{enumerate}
  The \emph{width} of a path decomposition $X_{1}, \dots, X_{r}$
  is $\max_{1 \le i \le r} |X_{i}| - 1$.
  The \emph{pathwidth} of $G$, denoted by $\pw(G)$,
  is the minimum width over all path decompositions of $G$.
  A path decomposition $X_{1}, \dots, X_{r}$ of $G$ is \emph{proper}
  if $\{i \mid u \in X_{i}\} \not\subset \{i \mid v \in X_{i}\}$ for all $u,v \in V(G)$.
  The \emph{proper pathwidth} of $G$, $\ppw(G)$, 
  is the minimum width over all proper path decompositions of $G$.
  Clearly, $\pw(G) \le \ppw(G)$ for any graph $G$.
  A nontrivial relation $\pw(G) \le \bw(G)$ is a corollary of the following fact.
  \begin{thm}
   [\cite{KS96}]
   For any graph $G$, $\bw(G) = \ppw(G)$.
  \end{thm}

  Let $\Bndry_{G}(A)$ denote the set of boundary vertices of $A$, that is,
  $\Bndry_{G}(A) = \{v \in V(G) \setminus A \mid \exists u \in A, \{u,v\} \in E(G)\}$.
  Let $\nBndry_{G}(k) = \min_{S \subseteq V(G), \; |S| = k} |\Bndry_{G}(S)|$.
  The \emph{vertex isoperimetric problem (VIP)} on a graph $G$ for given $k$ is to find a vertex set $S \subseteq V(G)$
  such that $|S| = k$ and $|\Bndry_{G}(S)| = \nBndry_{G}(k)$.
  We define the \emph{vertex boundary width} of $G$ as $\vbw(G) = \max_{1 \le k \le |V(G)|} \nBndry_{G}(k)$.
  We often omit the subscript $G$ of $\Bndry_{G}$ and $\nBndry_{G}$ if the graph $G$ is clear from the context.
  The following theorem implies $\vbw(G) \le \pw(G)$ for any graph $G$.
  \begin{thm}
   [\cite{CK06tw}]
   For any graph $G$ and $1 \le k \le |V(G)|$, $\nBndry(k) \le \pw(G)$.
  \end{thm}

  From the above observations, we have the inequality $\vbw(G) \le \pw(G) \le \bw(G)$ for any graph $G$.
  Harper~\cite{Har66,Har04} showed that the equality also holds for some graphs.
  An ordering on $V(G)$ is \emph{isoperimetric} for $G$ if
  $|\Bndry(I_{k})| = \nBndry(k)$ and
  $I_{k} \cup \Bndry(I_{k}) = I_{k + |\Bndry(I_{k})|}$ for all $k$,
  where $I_{k}$ is the set of the first $k$ vertices of $V(G)$ in the ordering.
  \begin{thm}
   [\cite{Har04}]
   If a graph $G$ has an isoperimetric ordering on $V(G)$
   then $\vbw(G) = \bw(G)$.
  \end{thm}

  The observations in this subsection give the following corollary.
  \begin{cor}
   If a graph $G$ has an isoperimetric ordering on $V(G)$
   then $\vbw(G) = \pw(G) = \bw(G)$.
  \end{cor}

  \subsection{Grids and tori}
  The \emph{Cartesian product} of graphs $G$ and $H$, denoted by $G \cp H$,
  is the graph whose vertex set is $V(G) \times V(H)$
  and in which a vertex $(g,h)$ is adjacent to a vertex $(g',h')$ if and only if either
  $g = g'$ and $\{h, h'\} \in E(H)$, or $h = h'$ and $\{g, g'\} \in E(G)$.
  It is easy to see that the Cartesian product operation is associative and commutative up to isomorphism.
  We denote the Cartesian product of $d$ graphs $G_{1}, G_{2}, \dots, G_{d}$ by $\prod_{i=1}^{d} G_{i}$.
  Also, we denote $\prod_{i=1}^{d} G$ by $G^{d}$.

  For $n \ge 2$, a \emph{path} $P_{n}$ is a graph whose vertex set is $\{0,\dots,n-1\}$
  and edge set is $\{\{i, i + 1\} \mid 0 \le i \le n-2\}$.
  For $n_{1} \le \dots \le n_{d}$, the graph $\dGrid{d}$ is called a \emph{$d$-dimensional grid}.
  We call $P_{n}^{3}$ a \emph{cubic grid}.
  Let $v = (v_{1},\dots, v_{d})$ is a vertex of $\dGrid{d}$.
  Then the \emph{weight} of $v$ is defined as $\wei(v) = \sum_{i = 1}^{d} v_{i}$.
  For $n \ge 3$, a \emph{cycle} $C_{n}$ is a path with a wrap around edge,
  that is, $V(C_{n}) = V(P_{n})$ and $E(C_{n}) = \{\{n-1, 0\}\} \cup E(P_{n})$.
  For $n_{1} \le \dots \le n_{d}$, we call $\dTorus{d}$ a \emph{$d$-dimensional even torus}.

  Let $\unit{k} = (\unit{k}_{1}, \dots, \unit{k}_{d})$ be the $k$th unit vector in a $d$-dimensional space,
  that is, $\unit{k}_{k} = 1$ and $\unit{k}_{i} = 0$ for all $i \ne k$.
  For $v = (v_{1}, \dots, v_{d})$, we have $(v + \unit{k})_{k} = v_{k} + 1$ and $(v + \unit{k})_{i} = v_{i}$ for all $i \ne k$.
  It is easy to see that
  for $u,v \in V(\dGrid{d})$, $\{u,v\} \in E(\dGrid{d})$ if and only if
  there exists an index $k$ such that either $u = v + \unit{k}$ or $v = u + \unit{k}$.

  We define the \emph{simplicial order} $\ls$ on $V(\dGrid{d})$
  by setting $(u_{1}, \dots, u_{d}) \ls (v_{1}, \dots, v_{d})$
  if and only if either $\wei(u) < \wei(v)$, or $\wei(u) = \wei(v)$ and
  there exists an index $j$ such that $u_{j} > v_{j}$ and $u_{i} = v_{i}$ for all $i < j$.
  Intuitively, vertices are ordered in the simplicial order by increasing weight and
  anti-lexicographically with each weight class~\cite{WWD09}.
  For example, the vertices of $P_{2} \cp P_{3} \cp P_{3}$ are ordered as follows:
  $(0,0,0) \ls (1,0,0) \ls (0,1,0) \ls (0,0,1) \ls (1,1,0) \ls (1,0,1) \ls (0,2,0) \ls (0,1,1) \ls (0,0,2) \ls (1,2,0) \ls (1,1,1) \ls (1,0,2) \ls (0,2,1) \ls (0,1,2) \ls (1,2,1) \ls (1,1,2) \ls (0,2,2) \ls (1,2,2)$.
  We also define $\les$, $\gs$, and $\ges$ naturally.
  Moghadam~\cite{Mog83,Mog05} and Bollob\'{a}s and Leader~\cite{BI91a,BI91b}
  showed independently that the simplicial order is isoperimetric for grids.
  \begin{thm}
   [\cite{BI91a,BI91b,Mog83,Mog05}]
   Let $n_{1} \le \dots \le n_{d}$.
   Then the simplicial order on $V(\dGrid{d})$ is isoperimetric for $\dGrid{d}$.
  \end{thm}
  \begin{cor}
   $\vbw(\dGrid{d}) = \pw(\dGrid{d}) = \bw(\dGrid{d})$.
  \end{cor}
  Riordan~\cite{Rio98} showed that even tori have an isoperimetric order.
  Thus, we also have the following equivalence. 
  \begin{cor}
   $\vbw(\dTorus{d}) = \pw(\dTorus{d}) = \bw(\dTorus{d})$.
  \end{cor}
  However, we do not need to give a formal definition of Riordan's ordering.
  This is because of the equivalence of the problem on even tori and grids.
  Recently, Bezrukov and Leck~\cite{BL09} have proved that the VIP on $d$-dimensional even tori
  is equivalent to the VIP on some $2d$-dimensional grids.
  \begin{thm}
   [\cite{BL09}]
   Let $G = \dTorus{d}$ and $H = P_{2}^{d} \cp \dGrid{d}$.
   Then, $\nBndry_{G}(m) = \nBndry_{H}(m)$ for $1 \le m \le |V(G)|$.
  \end{thm}
  \begin{cor}
   \label{cor:pw_equivalence}
   $\vbw(\dTorus{d}) = \vbw(P_{2}^{d} \cp \dGrid{d})$.
  \end{cor}

  From the above observations, the problems of determining the bandwidth and the pathwidth of grids and tori
  are solvable by determining the vertex boundary width of grids.
  Thus, in what follows, we only consider the problems on grids,
  and we identify the three parameters of grids.

  Note that the problem of determining the vertex boundary width is not trivial
  even if an isoperimetric order is known.
  For example, Harper~\cite{Har66} showed that the simplicial order on hypercubes $Q_{d} = P_{2}^{d}$ is isoperimetric.
  He stated without proof that it can be shown by induction that $\vbw(Q_{d}) = \sum_{k = 0}^{d-1} \binom{k}{\floor{k/2}}$.
  In his recent book~\cite{Har04}, Harper gave an exercise to prove the above equation with a remark ``surprisingly difficult.''
  Recently, Wang, Wu, and Dumitrescu~\cite{WWD09} have given the first explicit proof of the equation.

  We have one more motivation to determine the vertex boundary width of grids.
  Bollob\'{a}s and Leader~\cite{BI91a} showed the following fact.\footnote{
  In their paper, the theorem was stated in a more general form.
  }
  \begin{lem}
   [\cite{BI91a}] 
   Let $G_{i}$ be a connected graph of $n_{i}$ vertices for $1 \le i \le d$.
   Then, $\vbw(\prod_{i=1}^{d} G_{i}) \ge \vbw(\dGrid{d})$.
  \end{lem}
  Hence, we can use $\vbw(\dGrid{d})$ as a general lower bound on the bandwidth and the pathwidth of
  any $d$-dimensional (connected) graphs.
  Note that given a connected graph $G$, its prime factors $G_{1}, \cdots, G_{d}$ 
  such that $G = \prod_{i=1}^{d} G_{i}$ can be determined in linear time~\cite{IP07}.


 \section{Grids with relatively large maximum factors}
 \label{sec:large}
 Although our main target is the three-dimensional case,
 we investigate arbitrary dimension cases here.
 We show that, for $n_{1} \le \dots \le n_{d}$, if $n_{d}$ is at least $\sum_{i=1}^{d-1} (n_{i} - 1)$
 then the vertex boundary width of $\dGrid{d}$ is $\prod_{i=1}^{d-1} n_{i}$.
 That is, we prove the following theorem.
 \begin{thm}
  \label{thm:grid_large_max}
  If $n_{1} \le \dots \le n_{d}$ and $\sum_{i=1}^{d-1} (n_{i} - 1) \le n_{d}$, then
  \[
  \bw\left(\dGrid{d}\right) = \pw\left(\dGrid{d}\right) = \prod_{i=1}^{d-1} n_{i}.
  \]
 \end{thm}
 Note that the condition $\sum_{i=1}^{d-1} (n_{i} - 1) \le n_{d}$ always holds
 for the two dimensional case.
 The following lemma implies the theorem.
 \begin{lem}
  $\vbw(\dGrid{d}) = \prod_{i=1}^{d-1} n_{i}$
  if $n_{1} \le \dots \le n_{d}$ and $\sum_{i=1}^{d-1} (n_{i} - 1) \le n_{d}$.
 \end{lem}
 \begin{proof}
  It is known that $\pw(G \cp P_{n}) \le |V(G)|$ for any graph $G$ (see \cite{Dje08}).
  The upper bound is a direct corollary of this fact.

  Let $v \in V(\dGrid{d})$ be the first vertex of weight $\sum_{i=1}^{d-1} (n_{i} -1)-1$ in $\ls$.
  That is, $v_{i} = n_{i}-1$ for $1 \le i \le d-2$,
  $v_{d-1} = n_{d-1} - 2$, and $v_{d} = 0$.
  Let $S_{\les v} = \{u \in V(\dGrid{d}) \mid u \les v\}$.
  It suffices to show that $|\Bndry(S_{\les v})| \ge \prod_{i=1}^{d-1} n_{i}$.
  Since $v$ is the first vertex of weight $\wei(v)$,
  the set $S_{\les v}$ consists of $v$ and the vertices of weight at most $\wei(v)-1$.
  Thus, the vertices of weight $\wei(v)$ other than $v$ are in $\Bndry(S_{\les v})$.
  It is easy to see that $v$ has two neighbors of weight $\wei(v) + 1$,
  that is, $\{v + \unit{i} \mid i = d-1, d\}$.
  Hence, we have
  \begin{align*}
   |\Bndry(S_{\les v})|
   &= \left|\left\{u \in V\left(\dGrid{d}\right) \mid \wei(u) = \wei(v), u \ne v\right\} \cup \{v + \unit{i} \mid i = d-1, d\} \right|
   \\
   &= \left|\left\{u \in V\left(\dGrid{d}\right) \mid \wei(u) = \wei(v)\right\}\right| + 1.
  \end{align*}
  Let $L = \left\{u \in V\left(\dGrid{d}\right) \mid \wei(u) = \wei(v)\right\}$.
  Then, $|\Bndry(S_{\les v})| = |L| + 1$.

  We shall present a bijection between $L$ and $V(\dGrid{d-1}) \setminus \{z\}$,
  where $z$ is the last vertex of $\dGrid{d-1}$, that is,
  $z$ is the only vertex of weight $\sum_{i=1}^{d-1} (n_{i}-1)$ in $\dGrid{d-1}$.
  This implies that $|L| = |V(\dGrid{d-1})| - 1 = \prod_{i=1}^{d-1} n_{i} - 1$, as required.
  For $u \in V(\dGrid{d-1}) \setminus \{z\}$, we define $u'$ as follows:
  \[
   u'_{i} =
   \begin{cases}
    u_{i} & \IF 1 \le i \le d-1, \\
    \wei(v) - \wei(u) & \IF i = d.
   \end{cases}
  \]
  Obviously, $u' \in L$ implies $u \in V(\dGrid{d-1}) \setminus \{z\}$.
  We show that $u' \in L$ if $u \in V(\dGrid{d-1}) \setminus \{z\}$.
  It suffices to show that $0 \le u_{d}' = \wei(v) - \wei(u) \le n_{d}-1$.
  Since $\wei(u) \le \sum_{i=1}^{d-1} (n_{i}-1) - 1 = \wei(v)$, $\wei(v) - \wei(u) \ge 0$.
  Since $\sum_{i=1}^{d-1} (n_{i}-1) \le n_{d}$, $\wei(v) \le n_{d} - 1$, and so,
  $\wei(v) - \wei(u) \le n_{d} - 1$.
 \end{proof}
 As a corollary, we also have the following theorem for multi-dimensional even tori with relatively large maximum factors.
 \begin{thm}
  If $n_{1} \le \dots \le n_{d}$ and $\sum_{i=1}^{d-1} n_{i} \le n_{d} - 1$, then
  \[
  \bw\left(\dTorus{d}\right) = \pw\left(\dTorus{d}\right)= 2^{d} \prod_{i=1}^{d-1} n_{i}.
  \]
 \end{thm}
 \begin{proof}
  Since $\sum_{i=1}^{d-1} n_{i} \le n_{d} - 1$, we have that $\sum_{i=1}^{d} (2-1) + \sum_{i=1}^{d-1} (n_{i} - 1) \le n_{d}$.
  Therefore, \corname~\ref{cor:pw_equivalence} and \theoremname~\ref{thm:grid_large_max} imply that
  $\vbw\left(\dTorus{d}\right) = \vbw\left(P_{2}^{d} \cp \dGrid{d}\right) = 2^{d} \prod_{i=1}^{d-1} n_{i}$.
 \end{proof}


 \section{Three-dimensional grids}
 \label{sec:3gri}
 In this section, we concentrate to the three-dimensional case.
 Thus, we define $\Bndry := \Bndry_{\dGrid{3}}$ and $\nBndry := \nBndry_{\dGrid{3}}$.
 In 1974, FitzGerald~\cite{Fit74} determined the bandwidth of cubic grids.
 \begin{thm}
  [\cite{Fit74}]
  $\bw(P_{n}^{3}) = \floor{(3n^{2} + 2n)/4}$.
 \end{thm}
 We generalize the above result to the noncubic cases.
 More precisely, we prove the following theorem in this section.
 \begin{thm}
  \label{thm:3d}
  For $n_{1} \le n_{2} \le n_{3}$,
  \[
  \bw\left(\dGrid{3}\right) = \pw\left(\dGrid{3}\right) =
   \begin{cases}
    n_{1} n_{2} & \IF n_{1} + n_{2} - 2 \le n_{3}, \\
    \displaystyle
    n_{1} n_{2} - \floor{\left(\frac{n_{1} + n_{2} -n_{3} - 1}{2}\right)^{2}} & \OW.
   \end{cases}
  \]
 \end{thm}
 \theoremname~\ref{thm:grid_large_max} implies the first case in the above theorem.
 Thus, in the rest of this section, we can assume that $n_{1} + n_{2} -2> n_{3}$.
 Indeed, our actual assumption can be slightly weak. We assume $n_{1} + n_{2} - 2 \ge n_{3}$ instead.
 Thus the case of $n_{1} + n_{2} - 2 = n_{3}$ is proved again.
 Note that $\floor{\left({n_{1} + n_{2} -n_{3} - 1}\right)^{2}/4} = 0$ if $n_{1} + n_{2} - 2 = n_{3}$.
 
 Let $I_{s}$ be the set of the first $s$ vertices of $\dGrid{3}$ in $\ls$.
 The following simple fact is useful to determine the peak of the function $\nBndry$.
 \begin{fact}
  $\nBndry(s) = \nBndry(s-1) + |\Bndry(I_{s}) \setminus \Bndry(I_{s-1})| - 1$.
 \end{fact}
 \begin{proof}
  Clearly, $\nBndry(s) = \nBndry(s-1) + |\Bndry(I_{s}) \setminus \Bndry(I_{s-1})| - |\Bndry(I_{s-1}) \setminus \Bndry(I_{s})|$.
  From the definition of $\ls$, $|\Bndry(I_{s-1}) \setminus \Bndry(I_{s})| = |\{v\}| = 1$, where $v$ is the $s$th vertex in $\ls$.
 \end{proof}

 From the above fact, we can show that the function $\nBndry$ is non-decreasing for small weight classes,
 and non-increasing for large weight classes.
 Note that if $I_{s} = I_{s-1} \cup \{v\}$ then $\Bndry(I_{s}) \setminus \Bndry(I_{s-1}) \subseteq \{v + \unit{i} \mid 1 \le i \le 3\}$,
 that is, new boundary vertices are adjacent to the new vertex.
 \begin{lem}
  \label{lem:non_dec_inc}
  Let $v$ be the $s$th vertex. Then,
  \begin{enumerate}
   \item $\nBndry(s) \ge \nBndry(s-1)$ if $\wei(v) < n_{3}-1$, and
   \item $\nBndry(s) \le \nBndry(s-1)$ if $\wei(v) \ge n_{1} + n_{2} - 2$.
  \end{enumerate}
 \end{lem}
 \begin{proof}
  (1)
  We show that $v + \unit{3} \in \Bndry(I_{s}) \setminus \Bndry(I_{s-1})$.
  Clearly, $v_{3} < n_{3} - 1$, and so, $v + \unit{3} \in \Bndry(I_{s})$.
  Suppose  $v + \unit{3} \in \Bndry(I_{s-1})$.
  Then, $v - \unit{k} + \unit{3} \in I_{s-1}$ for some $k \in \{1,2\}$.
  This contradicts $v \ls v - \unit{k} + \unit{3}$.

  (2)
  We show that $\Bndry(I_{s}) \setminus \Bndry(I_{s-1}) \subseteq \{v + \unit{3}\}$, which implies $|\Bndry(I_{s}) \setminus \Bndry(I_{s-1})| \le 1$.
  Suppose $v + \unit{k} \in \Bndry(I_{s}) \setminus \Bndry(I_{s-1})$ for some $k \in \{1,2\}$.
  Then, $v_{k} < n_{k}-1$ and the vertices $\{v + \unit{k} - \unit{i} \mid k < i \le 3\}$ are not in $I_{s-1}$.
  This implies that $v_{i} = 0$ for $k < i \le 3$.
  So, we have that $\wei(v) = \sum_{i = 1}^{k} v_{i} < \sum_{i = 1}^{k} (n_{i} - 1)$,
  which implies that $\sum_{i=1}^{2} (n_{i} - 1) < \sum_{i = 1}^{k} (n_{i} - 1)$, a contradiction.
 \end{proof}

 For each weight class, we have the following similar property of $\nBndry$.
 \begin{lem}
  \label{lem:local_non_dec_inc}
  Let $v$ be the $s$th vertex and $n_{3} -1 \le \wei(v) < n_{1} + n_{2} - 2$. Then,
  \begin{enumerate}
   \item $\nBndry(s) \ge \nBndry(s-1)$ if $v \les (\wei(v) - n_{2} + 2, n_{2} - 2,0)$,
   \item $\nBndry(s) \le \nBndry(s-1)$ otherwise.
  \end{enumerate}
 \end{lem}
 \begin{proof}
  First we show that $(\wei(v) - n_{2} + 2, n_{2} - 2,0) \in \dGrid{3}$.
  It suffices to show that $0 \le \wei(v) - n_{2} + 2 \le n_{1} -1$.
  If $\wei(v) - n_{2} + 2 < 0$, then $n_{3} -1 \le \wei(v)  < n_{2} - 2$.
  This implies $n_{3} < n_{2} - 1$, a contradiction.
  If $n_{1}-1 < \wei(v) - n_{2} + 2$, then we have $n_{1} + n_{2} - 3 < \wei(v) \le n_{1} + n_{2} - 3$,
  which is a contradiction.
  
  (1)
  We show that $v + \unit{3} \in \Bndry(I_{s}) \setminus \Bndry(I_{s-1})$.
  Since $v \les (\wei(v) - n_{2} + 2, n_{2} - 2,0)$, $v_{1} \ge \wei(v) - n_{2} + 2$.
  Hence, $v_{2} + v_{3} \le n_{2}-2 < n_{3}-1$,
  and thus, $v + \unit{3} \in V(\dGrid{3})$.
  Suppose  $v + \unit{3} \in \Bndry(I_{s-1})$.
  Then, $v - \unit{k} + \unit{3} \in I_{s-1}$ for some $k \in \{1,2\}$.
  This contradicts $v \ls v - \unit{k} + \unit{3}$.

  (2)
  We show that $\Bndry(I_{s}) \setminus \Bndry(I_{s-1}) \subseteq \{v + \unit{3}\}$,
  which implies $|\Bndry(I_{s}) \setminus \Bndry(I_{s-1})| \le 1$.
  Suppose $v + \unit{k} \in \Bndry(I_{s}) \setminus \Bndry(I_{s-1})$ for some $k \in \{1,2\}$.
  Then, $v_{k} < n_{k}-1$ and the vertices $\{v + \unit{k} - \unit{i} \mid k < i \le 3\}$ are not in $I_{s-1}$.
  This implies that $v_{i} = 0$ for $k < i \le 3$.
  If $k=1$, then $v_{1} < n_{1} - 1$ and $v_{2} = v_{3} = 0$.
  Thus, we have that $v_{1} = \wei(v) \ge n_{3} - 1 \ge n_{1} - 1$, a contradiction.
  If $k=2$, then $v_{2} < n_{2} - 1$ and $v_{3} = 0$.
  Thus, $v_{1} = \wei(v) - v_{2} \ge \wei(v) - n_{2} + 2$.
  On the other hand, $v_{1} \le \wei(v)-n_{2}+2$ since $v \gs (\wei(v)-n_{2}+2,n_{2}-2,0)$.
  Therefore, we have $v = (\wei(v)-n_{2}+2, n_{2}-2, 0)$, a contradiction.
 \end{proof}

 \lemname{s}~\ref{lem:non_dec_inc} and \ref{lem:local_non_dec_inc} together imply the following corollary.
 \begin{cor}
  \label{cor:3d_peak}
  If $n_{3} \le n_{1} + n_{2} - 2$ then
  \[
  \vbw\left(\dGrid{3}\right)
  = \max_{r = n_{3} - 1}^{n_{1} + n_{2}  - 3} 
  \left|\Bndry\left(\left\{u \in V\left(\textstyle\dGrid{3}\right) \mid u \les (r-n_{2}+2, n_{2}-2,0)\right\}\right)\right|.
  \]
 \end{cor}

 From the above corollary, we can show the main result.
 Since $\vbw(P_{2}^{d})$ is known~\cite{Har66,WWD09},
 we assume $n_{3} \ge 3$.
 The assumptions $n_{3} \ge 3$ and $n_{3} \le n_{1} + n_{2} - 2$ imply $n_{2} \ge 3$.
 \begin{lem}
  If $n_{3} \le n_{1} + n_{2} - 2$ then $\vbw(\dGrid{3}) = n_{1}n_{2} - \floor{(n_{1} + n_{2} - n_{3} - 1)^{2}/4}$.
 \end{lem}
 \begin{proof}
  Let $S_{r} = \{v \in V(\dGrid{3}) \mid v \les (r-n_{2}+2, n_{2}-2,0)\}$.
  From \corname~\ref{cor:3d_peak}, it is sufficient to show that
  $\max_{r = n_{3}-1}^{n_{1} + n_{2} - 3}|\Bndry(S_{r})| = n_{1}n_{2} - \floor{(n_{1} + n_{2} - n_{3} - 1)^{2}/4}$.
  Assume $n_{3}-1 \le r < n_{1} + n_{2} - 2$.
  First we show that
  \[
  |\Bndry(S_{r})| =
  n_{1} n_{2} - \frac{(n_{1}+n_{2}-n_{3} -1)^{2} - 1}{4}
  - \frac{(2r - n_{1} - n_{2} - n_{3} + 4)^{2}}{4}.
  \]
  Let $B_{i} = \{v \in \Bndry(S_{r}) \mid \wei(v) = i\}$.
  Then, from the definition of $\ls$,
  $\Bndry(S_{r}) = B_{r} \cup B_{r+1}$ and
  \begin{align*}
   B_{r} &= \{v \in V({\textstyle\dGrid{3}}) \mid (r-n_{2}+2, n_{2}-3, 1) \les v \les (0, r - n_{3} + 1, n_{3}-1)\},
   \\
   B_{r +1} &= \{v \in V({\textstyle\dGrid{3}}) \mid (n_{1} - 1, r - n_{1} + 2, 0) \les v \les (r-n_{2}+2, n_{2}-2,1) \}.
  \end{align*}
  It is easy to see that 
  the four vertices $(r-n_{2}+2, n_{2}-3, 1)$, $(0, r - n_{3} + 1, n_{3}-1)$, $(n_{1} - 1, r - n_{1} + 2, 0)$, and $(r-n_{2}+2, n_{2}-2,1)$
  are in $V(\dGrid{3})$.
  To see this, use the assumptions $n_{1} \le n_{2} \le n_{3}$ and $n_{3}-1 \le r < n_{1} + n_{2} - 2$.

  Let $B_{i}(j)$ denote the set $\{v \in B_{i} \mid v_{1} = j\}$.
  Then, $B_{r} = \bigcup_{j=0}^{r-n_{2}+2} B_{r}(j)$ and $B_{r+1} = \bigcup_{j=r-n_{2}+2}^{n_{1}-1} B_{r+1}(j)$.
  \begin{claim}
   $|B_{r}| = (n_{2} - 2) + \sum_{j = r-n_{3}+2}^{r-n_{2}+1} n_{2} + \sum_{j = 0}^{r-n_{3}+1} (n_{2} + n_{3} - r + j - 1)$.
  \end{claim}
  \begin{proof}
   It is easy to see that
   \begin{align*}
    |B_{r}(r-n_{2}+2)|
    &=
    |\{(r - n_{2} + 2, a, b) \in V(\textstyle{\dGrid{3}}) \mid a + b = n_{2} - 2, a \le n_{2}-3\}|
    \\&=
    |\{(n_{2}-3, 1), (n_{2}-4, 2), \dots, (0, n_{2}-2)\}|
    =
    n_{2}-2.
   \end{align*}
   Assume that $r - n_{3}+2 \le j \le r-n_{2}+1$.
   Since $r-(r-n_{2}+1)-(n_{2}-1) = 0$ and $r-(r - n_{3}+2)-0 < n_{3}-1$,
   we have $0 \le r - j - k \le n_{3}-1$ for $0 \le k \le n_{2}-1$.
   Thus,
   \begin{align*}
    |B_{r}(j)|
    &= 
    |\{(j, k, r - j - k) \mid 0 \le k \le n_{2}-1, 0 \le r - j - k \le n_{3}-1\}|
    \\&=
    |\{(k, r - j - k) \mid 0 \le k \le n_{2} - 1\}|
    =
    n_{2}.
   \end{align*}
   Assume that $0 \le j \le r - n_{3} + 1$.
   Since $r - (r - n_{3} + 1) - (n_{2}-1) \ge 0$
   and $r - j - (r - j - n_{3} + 1) =  n_{3} - 1$,
   we have $0 \le r - j - k \le n_{3}-1$ for $r - j - n_{3} + 1 \le k \le n_{2}-1$.
   Hence,
   \begin{align*}
    |B_{r}(j)|
    &= 
    |\{(j, k, r - j - k) \mid 0 \le k \le n_{2}-1, 0 \le r - j - k \le n_{3}-1\}|
    \\&= 
    |\{(k, r - j - k) \mid r - j - n_{3} + 1 \le k \le n_{2}-1\}|
    = 
    n_{2} + n_{3} - r + j - 1.
   \end{align*}
   Thus, the claim holds.
  \end{proof}
  \begin{claim}
   $|B_{r+1}| = 2 + \sum_{j = r-n_{2}+3}^{n_{1} - 1} (r-j+2)$.
  \end{claim}
  \begin{proof}
   Obviously, $B_{r+1}(r-n_{2}+2) = \{(r-n_{2}+2, n_{2}-1, 0), (r-n_{2}+2, n_{2}-2, 1)\}$.
   Assume that $r - n_{2} + 3 \le j \le n_{1} - 1$.
   Since $r - j - (r-j+1) + 1 = 0$ and $r - (r - n_{2} + 3) - 0 + 1 = n_{2} - 2$,
   we have $0 \le r - j - k + 1 \le n_{3} - 1$ for $0 \le k \le r-j+1$.
   Hence,
   \begin{align*}
    |B_{r+1}(j)|
    &= 
    |\{(j, k, r - j - k + 1) \mid 0 \le k \le n_{2}-1, 0 \le r - j - k \le n_{3}-1\}|
    \\&=
    |\{(k, r - j - k + 1) \mid 0 \le k \le r-j+1\}|
    =
    r-j+2.
   \end{align*}
   Thus, the claim holds.
  \end{proof}
  
  The above two claims imply that
  \begin{align*}
   |\Bndry(S_{r})| &= |B_{r}| + |B_{r+1}|
   \\&=
    n_{2}  + \sum_{j = r-n_{3}+2}^{r-n_{2}+1} n_{2} + \sum_{j = 0}^{r-n_{3}+1} (n_{2} + n_{3} - r + j - 1)
    + \sum_{j = r-n_{2}+3}^{n_{1} - 1} (r-j+2)
   \\&=
   -r^{2} + (n_{1} + n_{2} + n_{3} -4)r - \frac{n_{1}^{2} + n_{2}^{2} + n_{3}^{2} - 5n_{1} - 5 n_{2} - 3n_{3} + 8}{2}
   \\&=
   \left(n_{1} n_{2} - \frac{(n_{1} + n_{2} - n_{3} -1)^{2} - 1}{4}\right)
   - \frac{(2r - n_{1} - n_{2} - n_{3} + 4)^{2}}{4}.
  \end{align*}
  Now, let $g(r) = (2r - n_{1} - n_{2} - n_{3} + 4)^{2}$.
  From the above observation,
  \begin{align*}
   \vbw\left(\dGrid{3}\right)
   &=
   \left(n_{1} n_{2} - \frac{(n_{1} + n_{2} - n_{3} -1)^{2} - 1}{4}\right)
   - \min_{r = n_{3} - 1}^{n_{1} + n_{2} -2} g(r)/4.
  \end{align*}
  Thus, minimizing $g(r)$, we can determine $\vbw\left(\dGrid{3}\right)$.
  \begin{claim}
   For $n_{3}-1 \le r \le n_{1} + n_{2} - 2$,
   $g(r)$ is minimized at $r = \floor{(n_{1} + n_{2} + n_{3} -4)/2}$.
  \end{claim}
  \begin{proof}
   Since $n_{3} \le n_{1} + n_{2} -2$, we have $n_{3}-1 \le \floor{(n_{1} + n_{2} + n_{3} -4)/2} \le n_{1} + n_{2} - 2$.
   Clearly,
   \[
   g(\floor {(n_{1} + n_{2} + n_{3} -4)/2}) = 
   \begin{cases}
    0 & \IF n_{1} + n_{2} + n_{3} \textrm{ is even}, \\
    1 & \OW.
   \end{cases}
   \]
   Since $r$ , $n_{1}$, $n_{2}$, and $n_{3}$ are integers and $g(r) = (2r - n_{1} - n_{2} - n_{3} + 4)^{2}$,
   $g(r)$ is a nonnegative integer.
   It is easy to see that if $g(x) = 0$ for some integer $x$, then $n_{1} + n_{2} + n_{3}$ is even.
   Since $g(r)$ is the square of some integer, the claim holds.
  \end{proof}
  Therefore, if $n_{3} \le n_{1} + n_{2} - 2$, then
  \begin{align*}
   \vbw\left(\dGrid{3}\right)
   &=
   n_{1}n_{2} - \frac{(n_{1}+n_{2}-n_{3} -1)^{2} - 1}{4}
   - 
   \begin{cases}
    0 & \IF n_{1} + n_{2} + n_{3} \textrm{ is even} \\
    1/4 & \OW
   \end{cases}
   \\&= 
   n_{1} n_{2} - \floor{\frac{(n_{1} + n_{2} - n_{3} -1)^{2}}{4}}.
  \end{align*}
  This completes the proof.
  (Note that $n_{1} + n_{2} + n_{3} \equiv n_{1} + n_{2} - n_{3} \pmod{2}.$)
 \end{proof}

 \section{Concluding remarks}
 \label{sec:conclusion}
 We have determined the vertex boundary width of three-dimensional grids.
 Since the vertex boundary width is equal to the bandwidth and the pathwidth for grids,
 the result properly extends some known results~\cite{Fit74,Chv75,EW08}.
 Since our result determines the bandwidth and the pathwidth of any grid whose dimension is three,
 it would be natural to study these parameters of four or more-dimensional grids.
 Here, we give a conjecture which was verified by computational experiments for $n \le 100$.
 \begin{conj}
  $\vbw(P_{n}^{4}) = \floor {(8n^{3} + 3n^{2} + 4n)/12}$.
 \end{conj}
 

 \section*{Acknowledgments}
 The authors are grateful to Koichi Yamazaki for careful reading of the manuscript.
 The authors would like to thank Hosien S.\ Moghadam
 for sending some recent papers including the reference~\cite{Mog05} to the first author's supervisor Koichi Yamazaki.
 The first author was supported by JSPS Research Fellowship for Young Scientists.


 

\end{document}